\documentclass[12pt,twoside,a4paper]{article}
\usepackage[utf8]{inputenc}
\usepackage[T1]{fontenc}
\usepackage{amsmath,amssymb,amsthm}
\usepackage{pdfpages}
\usepackage{tikz-cd}
\usepackage{color}
\usepackage{hyperref}

\newcommand{\HorD}{\mathcal{H}}
\newcommand{\VertD}{\mathcal{V}}
\newcommand{\pHor}{\pi_\HorD}
\newcommand{\pVert}{\pi_\VertD}
\newcommand{\Reeb}{\mathcal{R}}

\newcommand{\lsharp}{\sharp_\Lambda}

\newtheorem{theorem}{Theorem}
\newtheorem{proposition}{Proposition}

\newtheorem{lemma}{Lemma}
\theoremstyle{definition}
\newtheorem{definition}{Definition}

\theoremstyle{remark}
\newtheorem{remark}{Remark}

\newcommand*{\lieD}[1]{\mathcal{L}_{#1}}

\newcommand{\fracpartial}[2]{\frac{\partial {#1}}{\partial {#2}}}

\usepackage{color}
\usepackage[all]{xy}
\usepackage{graphicx}

\title{The Hamilton--Jacobi theory for contact Hamiltonian systems}
\author{
    {\bf\large Manuel de León}\textsuperscript{1,2},
    {\bf\large Manuel Lainz}\textsuperscript{1} and
    {\bf\large \'Alvaro Mu\~niz--Brea}\textsuperscript{1}
}

\begin{document}

\maketitle

\centerline{\textsuperscript{1}Instituto de Ciencias Matem\'aticas}
\centerline{Consejo Superior de Investigaciones Cient\'ificas}
\centerline{C/ Nicol\'as Cabrera, 13--15, 28049, Madrid. SPAIN}
\vskip 0.5cm
\centerline{\textsuperscript{2}Real Academia de Ciencias}
\centerline{C/ Valverde 22, 28004 Madrid. SPAIN}
\vskip 0.5cm

\begin{abstract}
The aim of this paper is to develop a Hamilton--Jacobi theory for
contact Hamiltonian systems. We find several forms for a suitable Hamilton-Jacobi equation
accordingly to the Hamiltonian and the evolution vector fields for a given Hamiltonian function.
We also analyze the corresponding formulation on the symplectification
of the contact Hamiltonian system, and establish the relations between these two approaches.  In the last section, some examples are discussed.
\end{abstract}

\tableofcontents

\section{Introduction}

The Hamilton–Jacobi equation is an alternative formulation of classical mechanics, equivalent to other formulations such as Lagrangian and Hamiltonian mechanics. The Hamilton–Jacobi equation is particularly useful in identifying conserved quantities for mechanical systems, which may be possible even when the mechanical problem itself cannot be solved completely. 

The Hamilton-Jacobi equation has been extensively studied in the case of symplectic Hamiltonian systems, more specifically, for Hamiltonian functions $H$ defined in the cotangent bundle $T^*Q$ of the configuration space $Q$. The Hamiltonian vector field is obtained by the equation
$$
i_{X_H} \, \omega_Q = dH
$$
where $\omega_Q$ is the canonical symplectic form on $T^*Q$. As we know, bundle coordinates $(q^i, p_i)$ are also
Darboux coordinates so that $X_H$ has the local form
$$
X_H = \frac{\partial H}{\partial p_i} \frac{\partial}{\partial q^i}
-  \frac{\partial H}{\partial q^i} \frac{\partial}{\partial p_i}
$$
The Hamilton-Jacobi problem consists in finding a function $S : Q \longrightarrow \mathbb{R}$ such that
\begin{equation}\label{hj}
H \left( q^i, \frac{\partial S}{\partial q^i} \right) = E,
\end{equation}
for some $E\in \mathbb{R}$.
The above equation (\ref{hj}) is called the Hamilton-Jacobi equation for $H$.
Of course, one easily see that (\ref{hj}) can be written as follows
\begin{equation}\label{hj2}
d (H \circ dS) = 0,
\end{equation}
which opens the possibility to consider general 1-forms on $Q$ (considered as sections of the cotangent bundle
$\pi_Q : T^*Q \longrightarrow Q$).

Recently, the observation that given such a section $\gamma : Q \longrightarrow T^*Q$ permits to relate 
$X_H$ with its projection $X_H^\gamma$ via $\gamma$ onto $Q$, in the sense that
$X_H^\gamma$ and $X_H$ are $\gamma$-related if and only if~\eqref{hj2} holds, provided that
$\gamma$ be closed (or, equivalently,  its image be a Lagrangian submanifold of $(T^*Q, \omega_Q)$)
has opened the possibility to discuss the Hamilton-Jacobi problem in many other
scenarios: nonholonomic systems, multisymplectic field theories, time-dependent mechanics, among others.

In \cite{dLS} we have started the extension of the Hamilton-Jacobi theory 
for contact Hamiltonian systems (see also \cite{dLVreview}). Let us recall that a contact Hamilton system is defined by
a Hamiltonian function on a contact manifold, in our case, the extended cotangent bundle $T^*Q \times \mathbb{R}$
equipped with the canonical contact form $\eta_Q = dz - \theta_Q$,
where $z$ is a global coordinate in $\mathbb{R}$ and $\theta_Q$ the Liouville form on $T^*Q$, with the obvious identifications. 

Contact Hamiltonian systems are widely used in many fields of Physics, like thermodynamics, dissipative systems,
cosmology, and even in Biology (the so-called neurogeometry).
The corresponding Hamilton equations were obtained in 1940 by G. Herglotz
using a variational principle that extends the usual one of Hamilton, but they can be alternatively derived using contact geometry.

The goal of this paper is to continue the study of the Hamilton-Jacobi problem
in the contact context, using the two vector fields associated to the Hamiltonian $H$:
\begin{itemize}
\item the Hamiltonian vector field
$$
X_H = \frac{\partial H}{\partial p_i} \frac{\partial}{\partial q^i} - 
\left(\frac{\partial H}{\partial q^i} + p_i \frac{\partial H}{\partial z} \right)\,  \frac{\partial}{\partial p_i} + 
\left(p_i \frac{\partial H}{\partial p_i} - H\right) \, \frac{\partial}{\partial z}
$$
\item the evolution vector field
$$
\mathcal{E}_H = \frac{\partial H}{\partial p_i} \frac{\partial}{\partial q^i} - 
\left(\frac{\partial H}{\partial q^i} + p_i \frac{\partial H}{\partial z} \right)\,  \frac{\partial}{\partial p_i} + 
p_i \frac{\partial H}{\partial p_i} \, \frac{\partial}{\partial z}
$$
\end{itemize}

We notice that the Hamilton-Jacobi problem has been treated by other authors~\cite{hj1,hj2}, who establish a relationship between the Herglotz variational principle and the Hamilton-Jacobi equation, although their interests are analytical rather than geometrical

The content of the paper is as follows. Section 2 is devoted to introduce the main ingredients of contact manifolds and
contact Hamiltonian systems as well as the interpretation of a contact manifold
as a Jacobi structure. In Section 3 we discuss the different types of submanifolds
of a contact manifold. Section 4 is the main part of the paper; there, we discuss
the Hamilton-Jacobi problem for a contact Hamiltonian vector field as well as for the 
corresponding evolution vector field. The results are more involved than
in the case of symplectic Hamiltonian systems due to the different possibilities
that may occur. In Section 5 we study the relations of the Hamilton-Jacobi problem
for a contact Hamiltonian systems and its symplectification. Finally, some examples are discussed in Section 6.

\section{Contact Hamiltonian systems}

\subsection{Contact manifolds}

Consider a contact manifold~\cite{dLV1,dLR,Bravetti2017} $(M, \eta)$ with contact form $\eta$; this means that
$\eta \wedge d\eta^n \not= 0$ and $M$ has odd dimension $2n+1$.
Then, there exists a unique vector field $\mathcal R$ (called Reeb vector field) such that
$$
i_{\mathcal R} \, d\eta = 0 \; , \; i_{\mathcal R}\, \eta = 1
$$

There is a Darboux theorem for contact manifolds (see \cite{god,libermarle}) so that around each point in $M$ one can find local coordinates 
(called Darboux coordinates) $(q^i, p_i, z)$ 
such that
$$
 \eta = dz - p_i \, dq^i
$$
and we have
$$
\mathcal R = \frac{\partial}{\partial z}
$$

The contact structure  defines an isomorphism between tangent vectors and covectors. For each $x\in M$,
    \begin{eqnarray*}
            \bar{\flat}: T_x M &\to  &   T_x^* M\\
            v            &\mapsto & i_v d \eta  + \eta(v) \eta.
    \end{eqnarray*}

    Similarly, we obtain a vector bundle isomorphism
   $$
 TM  \longrightarrow T^*M 
         $$
    over $M$.

 We will also denote by $\bar{\flat} : \mathfrak{X}(M) \to \Omega^1 (M)$ the corresponding isomorphism of $C^{\infty}(M)$-modules between vector fields and $1$-forms over $M$; 
 $\sharp$ will denote the inverse of $\bar{\flat}$.

Therefore, we have that
    $$
        \bar{\flat}(\Reeb) = \eta,
$$
    so that, in this sense, $\Reeb$ is the dual object of $\eta$.

For a Hamiltonian function $H$ on $M$ we define the Hamiltonian vector field $X_H$ by
$$
\bar{\flat} (X_H) = dH - (\mathcal R (H) + H) \, \eta
$$

In Darboux coordinates we get this local expression

\begin{equation}\label{hcont2}
X_H = \frac{\partial H}{\partial p_i} \frac{\partial}{\partial q^i} - 
\left(\frac{\partial H}{\partial q^i} + p_i \frac{\partial H}{\partial z} \right)\,  \frac{\partial}{\partial p_i} + 
\left(p_i \frac{\partial H}{\partial p_i} - H\right) \, \frac{\partial}{\partial z}
\end{equation}
Therefore, an integral curve $(q^i(t), p_i(t), z(t))$ of $X_H$ satisfies the 
contact Hamilton equations

\begin{eqnarray}\label{hcont3}
\frac{dq^i}{dt} & = & \frac{\partial H}{\partial p_i} \\
\frac{dp_i}{dt} & = & - \left(\frac{\partial H}{\partial q^i} +  p_i \frac{\partial H}{\partial z}\right)\\
\frac{dz}{dt} & = & \left(p_i \frac{\partial H}{\partial p_i} - H\right)
\end{eqnarray}

In addition to the Hamiltonian vector field $X_H$ associated to a Hamiltonian function $H$, there is another
relevant vector field, called {\sl evolution vector field} defined by
$$
\mathcal{E}_H = X_ H + H\Reeb
$$ 
so that it reads in local coordinates as follows

\begin{equation}\label{eval}
\mathcal{E}_H = \frac{\partial H}{\partial p_i} \frac{\partial}{\partial q^i} - 
(\frac{\partial H}{\partial q^i} + p_i \frac{\partial H}{\partial z} )\,  \frac{\partial}{\partial p_i} + 
p_i \frac{\partial H}{\partial p_i} \, \frac{\partial}{\partial z}
\end{equation}

Consequently, the integral curves of $\mathcal{E}_H$ satisfy the differential equations
\begin{eqnarray}\label{eval2}
\frac{dq^i}{dt} & = & \frac{\partial H}{\partial p_i} \\
\frac{dp_i}{dt} & = & - \left(\frac{\partial H}{\partial q^i} +  p_i \frac{\partial H}{\partial z}\right)\\
\frac{dz}{dt} & = & p_i \frac{\partial H}{\partial p_i} 
\end{eqnarray}

\begin{remark}{\rm
The evolution vector field plays a relevant role in the geometric description
of thermodynamics (see \cite{royal,houches}).}
\end{remark}

    Given a contact $2n+1$ dimensional manifold $(M, \eta)$, we can consider the following distributions on $M$, that we will call \emph{vertical} and \emph{horizontal} distribution, respectively:
    \begin{eqnarray*}
        \HorD &= &\ker \eta, \\
        \VertD &= &\ker d \eta.
    \end{eqnarray*}

We have a Withney sum decomposition

    $$
        TM = \HorD \oplus \VertD,
  $$
   and, at each point $x\in M$:
    $$
        T_x M = \HorD_x \oplus \VertD_x.
    $$
    We will denote by $\pHor$ and $\pVert$ the projections onto these subspaces. 
We notice that $\dim \HorD = 2 n$ and $\dim \VertD = 1$, and that $(d\eta)_{|_{\HorD}}$ is non-degenerate
and $\VertD$ is generated by the Reeb vector field $\Reeb$.

\begin{definition}
  \begin{enumerate}
\item A diffeomorphism between two contact manifolds $F:(M,\eta)\to (N, \xi)$ is a \emph{contactomorphism} if
    $$
        F^*\xi = \eta.
   $$

  \item A diffeomorphism $F:(M,\eta)\to (N, \xi)$ is a \emph{conformal contactomorphism} if there exist a nowhere zero function $f\in C^\infty(M)$ such that 
   $$
            F^*\xi = f \eta.
   $$
    
\item A vector field $X \in \mathfrak{X} (M)$ is an \emph{infinitesimal contactomorphism} (respectively \emph{infinitesimal conformal contactomorphism}) if its flow $\phi_t$ consists of contactomorphisms (resp. \emph{conformal contactomorphisms}).
\end{enumerate}
\end{definition}

Therefore, we have

\begin{proposition}

\begin{enumerate}
  \item A vector field $X$ is an infinitesimal contactomorphism if and only if
$$
        {\cal L}_{X} \eta = 0.
  $$

\item    $X$ is an infinitesimal conformal contactomorphism if and only if there exists $g \in C^{\infty}(M)$ such that
$$
        {\cal L}_{X} \eta = g \eta.
$$

    In this case, we say that $(g, X)$ is an \emph{infinitesimal conformal contactomorphism}.
\end{enumerate}
\end{proposition}

If $(M, \eta)$ is a $(2n+1)$-dimensional contact manifold and take
Darboux coordinates $(q^1,\ldots,q^n,p_1,\ldots,p_n,z)$, then

$$
            \VertD = < \frac{\partial}{\partial z} > \; , \;
            \HorD  = <A_i, B^i>
$$
  where
\begin{eqnarray*}
        A_i &= &\frac{\partial}{\partial q^i} - p_i \frac{\partial}{\partial z}\\
        B^i &=& \frac{\partial}{\partial p_i}.
    \end{eqnarray*}

 $\{A_1, B^1, \ldots, A_n, B^n, \Reeb\}$ and $\{d q^1, d p_1, \ldots, d q^n, d p_n, \eta\}$ are dual basis.

   We also have
$$
        [A_i, B^i ] = -\Reeb
$$

\subsection{Contact manifolds as Jacobi structures}

\begin{definition}
    A Jacobi manifold~\cite{Kirillov,Lich,libermarle} is a triple $(M,\Lambda,E)$, where $\Lambda$ is a bivector field (a skew-symmetric contravariant 2-tensor field) and $E \in \mathfrak{X} (M)$ is a vector field, so that the following identities are satisfied:
   $$
        [\Lambda,\Lambda] = 2 E \wedge \Lambda \; , \;
        {\cal L}_{E} \Lambda = [E,\Lambda] = 0,
$$
    where $[\cdot,\cdot ]$ is the Schouten–Nijenhuis bracket.
\end{definition}

    Given a Jacobi manifold $(M,\Lambda,E)$, we define the \emph{Jacobi bracket}:
    
        \begin{eqnarray*}
           \{\cdot, \cdot\} : C^\infty(M) \times C^{\infty}(M) & \mapsto \mathbb{R}, \\
            (f,g) &\mapsto \{f,g\},
        \end{eqnarray*}
    
   \noindent where
    $$
        \{f,g\} = \Lambda(d f, dg) + f E(g) - g E (f).
    $$

This bracket is bilinear, antisymmetric, and satisfies the Jacobi identity. Furthermore, it fulfills the weak Leibniz rule:
$$
        \operatorname{supp}(\{f,g\}) \subseteq \operatorname{supp} (f) \cap \text{supp} (g).
 $$
    That is, $(C^\infty(M), \{\cdot,\cdot\})$ is a local Lie algebra in the sense of Kirillov. 

Conversely, given a local Lie algebra $(C^\infty(M), \{\cdot,\cdot\})$, we can find a Jacobi structure on $M$ such that the Jacobi bracket coincides with the algebra bracket.

\begin{remark}{\rm The weak Leibniz rule is equivalent to this identity:

$$
\{f, gh\} = g \{f, h\} + h \{f, g\} + gh E(f)
$$
}
\end{remark}

Given a contact manifold $(M,\eta)$ we can define the associated Jacobi structure $(M, \Lambda, E)$ by 
$$
    \Lambda(\alpha,\beta) = - d \eta (\sharp\alpha, \sharp\beta), \quad
    E = - \Reeb,
$$
where $\sharp = \bar{\flat}^{-1}$.  For an arbitrary function $f$ on $M$ we can prove that the Hamiltonian vector field $X_f$ with respect to the contact structure $\eta$ coincides with the one defined by its associated Jacobi structure, say
$$
X_f = \sharp_\Lambda (df) - f {\mathcal R}
$$
where $\sharp_\Lambda$ is the vector bundle morphism from tangent covectors to tangent vectors defined by $\Lambda$,  i.e.
$$
<\sharp_\Lambda (\alpha), \beta> = \Lambda (\alpha, \beta),
$$
for all covectors $\alpha$ and $\beta$.

\section{Submanifolds}

As in the case of symplectic manifolds, we can consider several interesting types of submanifolds of a contact manifold $(M,\eta)$. To define them, we will use the following notion of \emph{complement} for contact structures~\cite{dLV1}:

    Let $(M,\eta)$ be a contact manifold and $x\in M$. Let $\Delta_x\subset T_x M$ be a linear subspace. We define the \emph{contact complement} of $\Delta_x$
    $$
        {\Delta_x}^{\perp_{\Lambda}} = \lsharp({\Delta_x}^o),
    $$
    where 
    ${\Delta_x}^o = \{\alpha_{x}\in T_x^*M \mid \alpha_x(\Delta_x)=0\}$ is the annihilator.

    We extend this definition for distributions $\Delta\subseteq TM$ by taking the complement pointwise in each tangent space.

Here, $\Lambda$ is the associated 2-tensor according to the previous section.

\bigskip

 \begin{definition}
 Let $N\subseteq M$ be a submanifold. We say that $N$ is:
    \begin{itemize}
        \item \emph{Isotropic} if $TN\subseteq {TN}^{\perp_{\Lambda}}$.
        \item \emph{Coisotropic} if $TN\supseteq {TN}^{\perp_{\Lambda}}$.
        \item \emph{Legendrian} or \emph{Legendre} if $TN= {TN}^{\perp_{\Lambda}}$.
    \end{itemize}
\end{definition}

The coisotropic condition can be written in local coordinates as follows.
  
    Let $N\subseteq M$ be a $k$-dimensional manifold given locally by the zero set of functions 
$\phi_a:U\to \mathbb{R}$, with $a\in \{1, …, k\}$. 

    We have that
   $$
        {TN}^{\perp_{\lambda}} = <Z_a \; | \; a=1, \dots, k >
   $$
    where
 $$
        Z_a = \sharp_\Lambda (d \phi_a)
$$

Therefore, $N$ is coisotropic if and only if, $Z_a(\phi_b)=0$ for all  $a,b$. 

Notice that
\begin{equation}\label{za}
Z_a = \left(\frac{\partial \phi_a}{\partial q^i} + p_i \frac{\partial \phi_a}{\partial z}\right) \frac{\partial}{\partial p_i}
+ \frac{\partial \phi_a}{\partial p_i} \left(\frac{\partial}{\partial q^i} - p_i \frac{\partial}{\partial z}\right).
\end{equation}

According to (\ref{za}), we conclude that $N$ is coisotropic if and only if
\begin{equation}\label{za2}
\left(\frac{\partial \phi_a}{\partial q^i} + p_i \frac{\partial \phi_a}{\partial z}\right) \frac{\partial \phi_b}{\partial p_i}
+ \frac{\partial \phi_a}{\partial p_i} \left(\frac{\partial \phi_b}{\partial q^i} - p_i \frac{\partial \phi_b}{\partial z}\right) = 0,
\end{equation}
for all $a,b$.

Using the above results, one can easily prove the following characterization of a Legendrian submanifold.

\begin{proposition}
Let $(M, \eta)$ be a contact manifold of dimension $2n+1$. A submanifold $N$ of $M$ is Legendrian if and only if it is a maximal integral manifold of $\ker \eta$
(and then it has dimension $n$).
\end{proposition}

Consider a function $f : Q \times \mathbb{R}$ and let $\eta_Q = dz - \rho^* \theta_Q$ the canonical contact structure
on $T^*Q \times \mathbb{R}$. Here $\rho : T^*Q \times \mathbb{R} \longrightarrow T^*Q$ is the canonical projection, and
$\theta_Q$ is the canonical Liouville form on $T^*Q$. In bundle coordinates $(q^i, p_i, z)$, we have
$$
\eta_Q = dz - p_i \, dq^i
$$
so that  $(q^i, p_i, z)$ are Darboux coordinates.

We denote by $j^1f : Q \longrightarrow T^*Q \times \mathbb{R}$ the 1-jet of $f$, say
$$
j^1f (q^i) = \left(q^i, \frac{\partial f}{\partial q^i}, f(q^i)\right)
$$

Then, one immediately checks that $j^1f(Q)$ is a Legendrian submanifold of $(T^*Q \times \mathbb{R}, \eta_Q)$. Moreover, we have

\begin{proposition}\label{section_coisotropic}
A section $\gamma : Q \longrightarrow T^*Q \times \mathbb{R}$ of the canonical projection
$T^*Q \times \mathbb{R} \longrightarrow Q$ is a Legendrian submanifold of
$(T^*Q \times \mathbb{R}, \eta_Q)$ if and only if $\gamma$ is locally the 1-jet of a function
$f : Q \longrightarrow \mathbb{R}$.
\end{proposition}

\begin{remark}{\rm The above result is the natural extension of the well-known 
fact that a section $\sigma$ of the cotangent bundle $\pi_Q : T^*Q \longrightarrow Q$ is a Lagrangian submanifold
with respect to the canonical symplectic structure $\omega_Q = - d \theta_Q$ on $T^*Q$
if and only if $\sigma$ is a closed 1-form (and hence, locally exact).
}
\end{remark}

\section{The Hamilton-Jacobi equations}

\subsection{The Hamilton-Jacobi equations for a Hamiltonian vector field}

We consider the extended phase space $T^{*}Q\times \mathbb{R}$, and a Hamiltonian function
$H:T^{*}Q\times \mathbb{R} \rightarrow \mathbb{R}$ (see the diagram below). 
\[
\xymatrix{ T^{*}Q\times \mathbb{R}
\ar[dd]^{\rho} \ar[ddrr]^{z}\ar@/^2pc/[ddrr]^{H}\\
  &  & &\\
T^{*}Q &  & \mathbb{R}}
\]
Recall that we have local canonical coordinates $\{q^i,p_i,z\}, i=1,\dots,n$ such that
the one-form is $\eta_Q=dz-\rho^{*}\theta_Q$, $\theta_Q$ being the canonical 1-form on $T^*Q$, can be locally expressed as follows
\begin{equation}\label{contactoneform}
 \eta_Q=dz-\sum_{i=1}^n p_idq^i.
\end{equation}
$(T^{*}Q\times \mathbb{R},\eta)$ is a contact manifold with 
Reeb vector field 
$\mathcal{R}=\frac{\partial}{\partial z}.$

Consider the Hamiltonian vector field $X_H$ for a given Hamiltonian function, say
\begin{equation}
 X_H=\sharp_{\Lambda}(dH)+H \Reeb.
\end{equation}
In coordinates, it reads
{\begin{footnotesize}
\begin{equation}\label{1hvf}
 X_H= \sum_{i=1}^n\frac{\partial H}{\partial p_i}\frac{\partial}{\partial q^i} -\sum_{i=1}^n\left(p_i\frac{\partial H}{\partial z}+\frac{\partial H}{\partial q^i}\right)\frac{\partial}{\partial p_i} + \sum_{i=1}^n \; \left(p_i\frac{\partial H}{\partial p_i} -H\right)\frac{\partial}{\partial z} 
\end{equation}
\end{footnotesize}}
We also have
\begin{equation*}
 \bar{\flat}{(X_H)}= dH -(\mathcal{R}(H)+H)\eta,
\end{equation*}
where $\flat$ is the isomorphism previously defined. We also have that 
\begin{equation}\label{1exph}
 \eta(X_H)=-H.
\end{equation}
Recall that $(T^{*}Q\times \mathbb{R},\Lambda,\mathbb{R})$ is a Jacobi manifold with $\Lambda$ given in the usual way
(see Section 2.2).
The proposed contact structure provides us with the {\it contact Hamilton equations}.

\begin{equation}\label{hamileq}
\left\{\begin{aligned}
 {\dot q}^i&=\frac{\partial H}{\partial p_i},\\
 {\dot p}_i&=-\frac{\partial H}{\partial q^i}-p_i\frac{\partial H}{\partial z},\\
{\dot z}&=p_i\frac{\partial H}{\partial p_i}-H.
 \end{aligned}\right.
 \end{equation}
for all $i=1,\dots,n$.

Consider $\gamma$ a section of $\pi:T^{*}Q\times \mathbb{R} \rightarrow Q\times \mathbb{R}$, i.e., $\pi\circ \gamma=\text{id}_{Q\times \mathbb{R}}$. We can use $\gamma$ to project $X_H$ on $Q\times \mathbb{R}$
just defining a vector field $X_{H}^{\gamma}$ on $Q\times \mathbb{R}$ by
\begin{equation}\label{hjpar}
 X_H^{\gamma}=T_{\pi}\circ X_{H}\circ \gamma.
\end{equation}
The following diagram summarizes the above construction
\[
\xymatrix{ T^{*}Q\times \mathbb{R}
\ar[dd]^{\pi} \ar[rrr]^{X_H}&   & &T(T^{*}Q\times \mathbb{R})\ar[dd]^{T{\pi}}\\
  &  & &\\
Q\times \mathbb{R} \ar@/^2pc/[uu]^{\gamma}\ar[rrr]^{X^{\gamma}_H}&  & & T(Q\times \mathbb{R})}
\]

Assume that in local coordinates we have
$$
(q^i, z) \mapsto \gamma(q^i, z) = (q^i, \gamma_j(q^i, z), z)
$$
We can compute $T\gamma (X_H^\gamma)$ and obtain

\begin{equation}\label{congamma}
T\gamma (X_H^\gamma) = \frac{\partial H}{\partial p_i} \frac{\partial}{\partial q^i} +
\left(\frac{\partial H}{\partial p_i} \frac{\partial \gamma_j}{\partial q^i} +
\left(\gamma_i \frac{\partial H}{\partial p_i} -H\right) \frac{\partial \gamma_j}{\partial z}\right)  \frac{\partial}{\partial p_j}
+ \left(\gamma_i \frac{\partial H}{\partial p_i} - H\right) \frac{\partial}{\partial z}
\end{equation}

Therefore, from (\ref{1hvf}) and (\ref{congamma}), we have that
$$
X_H \circ \gamma = T\gamma (X_H ^\gamma)
$$
if and only if

\begin{equation}\label{hjlocal}
\frac{\partial H}{\partial q^j} + 
 \frac{\partial \gamma_j}{\partial q^i}  \frac{\partial H}{\partial p_i} + 
\gamma_j \frac{\partial H}{\partial z} + \gamma_i \frac{\partial \gamma_j}{\partial z} \frac{\partial H}{\partial p_i} - H \frac{\partial \gamma_j}{\partial z} = 0.
\end{equation}

Assume now that

\begin{enumerate}

\item $\gamma(Q\times \mathbb{R})$ is a coisotropic submanifold of 
$(T^{*}Q\times \mathbb{R}, \eta_Q)$;

\item $\gamma_z (Q)$ is a Lagrangian submanifold of $(T^{*}Q, \omega_Q)$, for any $z \in \mathbb{R}$,
where $\gamma_z (q) = \rho \circ \gamma (q, z)$.

Notice that the above two conditions imply that $\gamma(Q \times \mathbb{R})$ is foliated by
Lagrangian leaves $\gamma_z(Q)$, $z \in \mathbb{R}$.

\end{enumerate}

We will discuss the consequences of the above conditions.
The submanifold $\gamma(Q \times \mathbb{R})$ is locally defined by the functions
$$
\phi_i = p_i - \gamma_i = 0
$$
Therefore, the first condition is equivalent to

\begin{equation}\label{coiso}
\frac{\partial \gamma_i}{\partial q^j} - \gamma_j \frac{\partial \gamma_i}{\partial z} - \frac{\partial \gamma_j}{\partial q^i} + 
\gamma_i \frac{\partial \gamma_j}{\partial z} = 0
\end{equation}
If, in addition, $\gamma_z(Q)$ is Lagrangian submanifold for any fixed $z \in \mathbb{R}$, then we obtain
\begin{equation}\label{coiso2}
\frac{\partial \gamma_i}{\partial q^j} - \frac{\partial \gamma_j}{\partial q^i} = 0 
\end{equation}
and, using again (\ref{coiso}), we get
\begin{equation}\label{coiso3}
\gamma_j \frac{\partial \gamma_i}{\partial z} - 
\gamma_i \frac{\partial \gamma_j}{\partial z} = 0
\end{equation}

Under the above conditions (using \ref{coiso2} and \ref{coiso3}), \ref{hjlocal} becomes

\begin{equation}\label{hjlocal2}
\frac{\partial H}{\partial q^j} + 
\frac{\partial H}{\partial p_i} \frac{\partial \gamma_i}{\partial q^j} +
\gamma_j \left( \frac{\partial H}{\partial z} + \frac{\partial H}{\partial p_i} \frac{\partial \gamma_i}{\partial z} \right) - H \frac{\partial \gamma_j}{\partial z} = 0.
\end{equation}

We can write down eq (\ref{hjlocal2}) in a more friendly way. First of all, consider the following functions and 1-forms defined on
$Q \times \mathbb{R}$:

\begin{enumerate}

\item 
$$
\gamma_o =  \frac{\partial H}{\partial z} + \frac{\partial H}{\partial p_i} \frac{\partial \gamma_i}{\partial z} 
$$
\item 
$$
d(H \circ \gamma_z) = \left(\frac{\partial H}{\partial q^j} + 
\frac{\partial H}{\partial p_i} \frac{\partial \gamma_i}{\partial q^j} \right) dq^j
$$
\item
$$
i_{\frac{\partial}{\partial z}} (d(\gamma^* \theta_Q)) = \frac{\partial \gamma_j}{\partial z} dq^j
$$

\end{enumerate}

Therefore, eq (\ref{hjlocal2}) is equivalent to

\begin{equation}\label{hjglobal}
d (H \circ \gamma_z) + \gamma_o (\gamma^* \theta_Q) - (H\circ \gamma)  (i_{\frac{\partial}{\partial z}} (d(\gamma^* \theta_Q))) = 0.
\end{equation}

\begin{theorem}
 Assume that a section $\gamma$ of the projection $T^*Q \times \mathbb{R} \longrightarrow Q \times \mathbb{R}$
is such that $\gamma(Q\times \mathbb{R})$ is a coisotropic submanifold of 
$(T^{*}Q\times \mathbb{R}, \eta_Q)$, and $\gamma_z (Q)$ is a Lagrangian submanifold of $(T^{*}Q, \omega_Q)$, for any $z \in \mathbb{R}$. 
Then, the vector fields $X_H$ and $X_H^{\gamma}$ are $\gamma$-related if and only if (\ref{hjlocal2}) holds (equivalently, (\ref{hjglobal}) holds). 
\end{theorem}

Equations (\ref{hjlocal2}) and (\ref{hjglobal}) are indistinctly referred as a {\it Hamilton--Jacobi equation with respect to a contact structure}. A section $\gamma$ fullfilling the assumptions of the theorem and the Hamilton-Jacobi equation will
be called a {\it solution} of the Hamilton--Jacobi problem for $H$.

\begin{remark}{\rm
Notice that if $\gamma$ is a solution of the Hamilton--Jacobi problem for $H$, then
$X_H$ is tangent to the coisotropic submanifold $\gamma(Q \times \mathbb{R})$, but
not necesarily to the Lagrangian submanifolds $\gamma_z(Q)$, $z \in \mathbb{R}$. This occurs when
$$
X_H(z -z_0) = 0
$$
for any $z_0$, that is, if and only if
$$
H \circ \gamma_{z_0} = \gamma_i \frac{\partial H}{\partial p_i}
$$
In such a case, we call $\gamma$ an {\it strong solution} of the Hamilton--Jacobi problem.
}
\end{remark}

A characterization of conditions on the submanifolds $\gamma(TQ \times \mathbb{R}), \gamma_z(TQ)$ can be given as follows. Let $\sigma: Q \times \mathbb{R} \to \Lambda^k(T^* Q)$ be a $z$-dependent $k$-form on $Q$. Let $d_Q \sigma$ be the exterior derivative at fixed $z$, that is
\begin{equation}
    d_Q \sigma(q^i,z) = d \sigma_z(q^i),
\end{equation}
where $\sigma_z= \sigma(\cdot,z)$. In local coordinates, we have
\begin{equation}
    \begin{split}
        d_Q f &= \frac{\partial f}{\partial  q^i} d q^i,\\
        d_Q (\alpha_i d q^i) &= \frac{\partial \alpha_j}{\partial  q^i} d q^i \wedge d q^j,
    \end{split}
\end{equation}
where $f:Q \times \mathbb{R} \to \mathbb{R}$ is a function and $\alpha = \alpha_i dq^i: Q \times \mathbb{R} \to \Lambda^1(T^* Q)$ is a $z$-dependent $1$-form.

\begin{theorem}\label{thm:coisotropic_lagrangian_section}
    Let $\gamma$ be a section of $T^*Q \times \mathbb{R}$ over $Q \times \mathbb{R}$. Then $\gamma(Q \times \mathbb{R})$ is a coisotropic submanifold and $\gamma_{z_0}(TQ)$ are Lagrangian submanifolds for all $z_0$ if and only if $d_Q \gamma = 0$ and $\lieD{\partial/\partial z} \gamma = \sigma \gamma$ for some function $\sigma:Q\times \mathbb{R} \to \mathbb{R}$. That is, there exists locally a function $f:Q \times \mathbb{R} \to \mathbb{R}$ such that $d_Q f = \gamma$ and $d_Q \frac{\partial f}{\partial z} = \sigma d_Q f$.
\end{theorem}

\begin{proof}
    Fix $z_0 \in \mathbb{R}$, then, $\gamma_{z_0}(Q)$ is Lagrangian if and only if $\gamma_{z_0}$ is closed, hence $d \gamma_{z_0} =0$, so all $\gamma_{z_0}(Q)$ are Lagrangian if and only if $d_Q \gamma =  0$. By the Poincaré Lemma, locally $\gamma = d_Q f$,

    Now also assume that $\gamma(Q\times \mathbb{R})$ is coisotropic. Then, equation~\eqref{coiso3} can be written as
    \begin{equation}
        \gamma \wedge \lieD{\partial/\partial z}{\gamma} = 0,
    \end{equation} 
    or, equivalently, that $\gamma$ and $\lieD{\partial/\partial z}{\gamma}$ are proportional. 
    
    Locally, we obtain that $d_Q \frac{\partial f}{\partial z} = \sigma d_Q f$.

\end{proof}

\subsubsection{Complete solutions}

Next, we shall discuss the notion of complete solutions of the Hamilton--Jacobi problem for a Hamiltonian $H$.

\begin{definition}
 A {\it complete solution} of the Hamilton--Jacobi equation for a Hamiltonian $H$ 
 is a diffeomorphism $\Phi:Q\times \mathbb{R}\times \mathbb{R}^n\rightarrow T^{*}Q\times \mathbb{R}$ such that for any set of
 parameters $\lambda\in \mathbb{R}^n, \lambda=(\lambda_1,\dots,\lambda_n)$, the mapping
 
 \begin{equation}
 \begin{array}{ccc}
  \Phi_{\lambda}:Q\times \mathbb{R}& \rightarrow &  T^{*}Q\times \mathbb{R}  \\
  (q^i, z) &\mapsto &  \Phi_\lambda(q^i, z) = \Phi(q^i, z, \lambda)
 \end{array}
 \end{equation}
\noindent
is a solution of the Hamilton--Jacobi equation. If, in addition, any $\Phi_\lambda$ is strong,
then the complete solution is called an strong complete solution.
\end{definition}

We have the following diagram

\[
    \begin{tikzcd}
        Q \times \mathbb{R} \times \mathbb{R}^n \arrow[r, "\Phi", shift left] \arrow[d, "\alpha"] & T^*Q \times \mathbb{R} \arrow[l, "\Phi^{-1}", shift left] \arrow[d, "f_i"] \\
        \mathbb{R}^n \arrow[r, "\pi_i"]                                                           & \mathbb{R}                                                                
    \end{tikzcd}
\]
where we define functions $f_i$ such that for a point $p\in T^{*}Q\times \mathbb{R}$, it is satisfied
\begin{equation}\label{functions}
 f_i(p)=\pi_i\circ \alpha\circ \Phi^{-1}(p).
\end{equation}
and $\alpha:Q\times \mathbb{R}\times \mathbb{R}^n\rightarrow \mathbb{R}^n$ is the canonical projection.

The first immediate result is that
$$
\hbox{Im} \; \Phi_\lambda = \cap_{i=1}^n \, f_i^{-1}(\lambda_i)
$$
where $\lambda = (\lambda_1, \cdots, \lambda_n)$. In other words,
$$
\hbox{Im} \; \Phi_\lambda = \{ x \in T^*Q \times \mathbb{R} \; | \; f_i(x) = \lambda_i, i=1, \cdots, n\}
$$
Therefore, since $X_H$ is tangent to any of the submanifolds $\hbox{Im} \; \Phi_\lambda$, we deduce that
$$
X_H (f_i) = 0
$$
So, these functions are conserved quantities.

Moreover, we can compute
$$
\{f_i, f_j\} = \Lambda (df_i, df_j) - f_i \mathcal{R}(f_j) + f_j \mathcal{R}(f_i)
$$
But
$$
\Lambda (df_i, df_j) = \sharp_\Lambda( df_i)(f_j) = 0
$$
since $(T \hbox{Im} \Phi_\lambda)^\perp = \sharp_\Lambda ((T \hbox{Im} \Phi_\lambda)^o) \subset T \hbox{Im} \Phi_\lambda$, so 
\begin{equation}\label{involution}
\{f_i, f_j\} = - f_i \mathcal{R}(f_j) + f_j \mathcal{R}(f_i)
\end{equation}

\begin{theorem}
 There exist no linearly independent commuting set of first-integrals in involution \eqref{functions} for a complete strong solution of the Hamilton--Jacobi
 equation on a contact manifold.
\end{theorem}

{\bf Proof:} If all the particular solutions are strong, then the Reeb vector field $\mathcal{R}$ will be
transverse to the coisotropic submanifold $\Phi_\lambda(Q \times \mathbb{R})$. Indeed, 
if $\mathcal{R}$ is tangent to that submanifold, we would have
$$
\mathcal{R} (p_i - (\Phi_\lambda)_i) = - \frac{\partial (\Phi_\lambda)_i}{\partial z}
$$
where $\Phi_\lambda (q^i, z) = (q^i, (\Phi_\lambda)_i, z)$. So, $\Phi_\lambda$ does not depend on $z$, hence it cannot be a diffeomorphism.

Therefore, if the brackets $\{f_i, f_j\} $ vanish, then we woul obtain
that the functions $f_i$ cannot be linearly independent. Indeed, we should have
$$
f_i \mathcal{R}(f_j) =  f_j \mathcal{R}(f_i)
$$
for all $i, j$. But this would imply that $f_i$ and $f_j$ are linearly dependent in the case
$\lambda = (0, \dots, 0)$. 

$\hfill \Box$



\subsubsection{An alternative approach}

Instead of considering sections of $\pi : T^*Q \times \mathbb R \longrightarrow Q \times \mathbb{R}$ as above, 
we could consider a section of the canonical projection 
$\pi : T^*Q \times \mathbb R \longrightarrow Q$, say
$\gamma :  Q \to T^*Q \times \mathbb R$.

In local coordinates, we have
$$
(q^i) \mapsto \gamma(q^i) = (q^i, \gamma_j(q^i), \gamma_z(q^i))
$$

We want $\gamma$ to fulfill
\begin{equation}\label{eq:HJ}
X_H\circ\gamma=T\gamma \circ X_H^\gamma,
\end{equation}
where $X_H^\gamma=T\pi\circ X_H\circ\gamma$. Using the local expression of $X_H$ we have $X_H^\gamma=\sum_{i=1}^n
\left(\frac{\partial H}{\partial p_i}\circ\gamma\right) \frac{\partial}{\partial q^i}$, and since

$$
T\gamma \left(\fracpartial{}{q^i}\right)=\fracpartial{}{qi}+\sum_{j=1}^n\fracpartial{\gamma_j}{q_i}\fracpartial {}{p_j}+\fracpartial{\gamma_z}{q_i}\fracpartial{}{p_j}
$$
equation (\ref{eq:HJ}) holds if and only if:
\begin{eqnarray}\label{eq:HJ2}
%
&& -\left(\gamma_i\fracpartial {H}{z}+\fracpartial{H}{q_i}\right)=\sum_{j=1}^n\fracpartial {H}{p_j}\fracpartial{\gamma_i}{q^j},\quad i=1,\dots,n, \\
%
&& \sum_{i=1}^n\gamma_i\fracpartial{H}{p_i}-H=\sum_{i=1}^n\fracpartial {H}{p_i}\fracpartial{\gamma_z}{q^i}.
\end{eqnarray}

Now, notice that
$$
\tilde{\gamma} = \rho \circ \gamma
$$
is a 1-form on $Q$. Then, we locally have $\tilde{\gamma} = \gamma_i(q) \, dq^i$.

Next, we assume that $\gamma(Q)$ is a Legendrian submanifold of $(T^*Q \times \mathbb{R}, \eta_Q)$. This implies that $\tilde{\gamma}(Q)$ is a Lagrangian submanifold of $(T^*Q, \omega_Q)$.

By Proposition~\ref{section_coisotropic}, $\gamma(Q)$ is a Legendrian submanifold if and only if it is locally the 1-jet of a function, namely $\gamma=j^1 \gamma_z$, where we consider $\gamma_z$ as a function from $Q$ to $\mathbb R$. In other words, we have:
\begin{equation}\label{cond1}
\gamma_i = \frac{\partial \gamma_z}{\partial q^i}
\end{equation}

If we assume that the section $\gamma$ fulfills the above condition, we can see that equations (\ref{eq:HJ2}) become
\begin{eqnarray}
H \circ \gamma = 0. \label{eq:HJ3}
\end{eqnarray}

\begin{definition}
Assume that a section $\gamma$ such that $\gamma(Q)$ is a Legendrian submanifold of $(T^*Q \times \mathbb{R}, \eta_Q)$
and $\tilde{\gamma}(Q)$ is a Lagrangian submanifold of $(T^*Q, \omega_Q)$. Then $\gamma$ is called a solution of the Hamilton-Jacobi problem for the contact Hamiltonian $H$ if and if equation (\ref{eq:HJ3}) holds.
\end{definition}

We could discuss the existence of complete solutions in a similar manner to the case of the Hamiltonian vector field. We omit the details that are left to the reader.

\subsection{The Hamilton-Jacobi equations for the evolution vector field}

\subsubsection{A first approach}

Assume that $\mathcal{E}_H$ is the evolution vector field defined for a Hamiltonian function
$H: T^*Q \times \mathbb{R} \longrightarrow \mathbb{R}$. Then, we have

\begin{equation}\label{evol}
\mathcal{E}_H = \frac{\partial H}{\partial p_i} \frac{\partial}{\partial q^i} - 
\left(\frac{\partial H}{\partial q^i} + p_i \frac{\partial H}{\partial z} \right)\,  \frac{\partial}{\partial p_i} + 
p_i \frac{\partial H}{\partial p_i} \, \frac{\partial}{\partial z}
\end{equation}

Assume that $\gamma$ is a section of the canonical projection 
$\pi : T^*Q \times \mathbb R \longrightarrow Q \times \mathbb{R}$, say
$\gamma :  Q \times \mathbb{R} \to T^*Q \times \mathbb R$.

In local coordinates we have
$$
(q^i, z) \mapsto \gamma(q^i) = (q^i, \gamma_j(q^i), z)
$$

Therefore, we can define the projected evolution vector field
$$
\mathcal{E}_H^\gamma = T\pi \circ \mathcal{E}_H \circ \gamma.
$$

We have that $\mathcal{E}_H \circ \gamma = T\gamma(\mathcal{E}_H^\gamma)$ if and only if
\begin{equation}\label{nuevo}
\frac{\partial H}{\partial q^j} + \frac{\partial H}{\partial p_i} \frac{\partial \gamma_j}{\partial q^i}
+ \gamma_i  \frac{\partial H}{\partial p_i}  \frac{\partial \gamma_j}{\partial z} + \gamma_j \frac{\partial H}{\partial z} = 0
\end{equation}

Assume now that

\begin{enumerate}

\item $\gamma(Q\times \mathbb{R})$ is a coisotropic submanifold of 
$(T^{*}Q\times \mathbb{R}, \eta_Q)$;

\item $\gamma_z (Q)$ is a Legendrian submanifold of $(T^{*}Q\times \mathbb{R}, \eta_Q)$, for any $z \in \mathbb{R}$,
where $\gamma_z (q) = \gamma (q, z)$.

\end{enumerate}

Then, a direct computation shows that (\ref{nuevo}) becomes
\begin{equation}\label{nuevo2}
d( H \circ \gamma) + \gamma_o \, \gamma^*(\theta_Q) = 0 \; ,
\end{equation}
where 
$$
\gamma_o =  \frac{\partial H}{\partial z} + \frac{\partial H}{\partial p_i} \frac{\partial \gamma_i}{\partial z} 
$$

\begin{theorem}
 Assume that a section $\gamma$ of the projection $T^*Q \times \mathbb{R} \longrightarrow Q \times \mathbb{R}$
is such that $\gamma(Q\times \mathbb{R})$ is a coisotropic submanifold of 
$(T^{*}Q\times \mathbb{R}, \eta_Q)$, and $\gamma_z (Q)$ is a Legendrian submanifold of $(T^{*}Q\times \mathbb{R}, \eta_Q)$, for any $z \in \mathbb{R}$. 
Then, the vector fields $\mathcal{E}_H$ and $\mathcal{E}_H^{\gamma}$ are $\gamma$-related if and only if (\ref{nuevo2}) holds. 
\end{theorem}

Equation (\ref{nuevo2}) is referred as a {\it Hamilton--Jacobi equation for the evolution vector field}.
A section $\gamma$ fullfilling the assumptions of the theorem and the Hamilton-Jacobi equation will
be called a {\it solution} of the Hamilton--Jacobi problem for the evolution vector field of $H$.

\subsubsection{An alternative approach}\label{sec:evolution_alt}

We will maintain the notations of the previous subsection, but now $\gamma$ is a section
of the canonical projection 
$\pi : T^*Q \times \mathbb R \longrightarrow Q$, say
$\gamma :  Q \to T^*Q \times \mathbb R$.

In local coordinates we have
$$
(q^i) \mapsto \gamma(q^i) = (q^i, \gamma_j(q^i), \gamma_z(q^i))
$$

As in the above sections, we define the projected evolution vector field
$$
\mathcal{E}_H^\gamma = T\pi \circ \mathcal{E}_H \circ \gamma.
$$

A direct computation shows that
$\mathcal{E}_H \circ \gamma = T\gamma(\mathcal{E}_H^\gamma)$ if and only if
\begin{eqnarray}
&&\frac{\partial H}{\partial q^j} + \frac{\partial H}{\partial p_i} \frac{\partial \gamma_j}{\partial q^i} \label{11a}
+ \gamma_j \frac{\partial H}{\partial z} = 0  \label{11b} \\
 &&  \frac{\partial H}{\partial p_i} \left(\frac{\partial \gamma_z}{\partial q^i} - \gamma_i \right) = 0 \label{12}
\end{eqnarray}

If we assume that $\gamma = j^1 f$, for some function $f : Q \longrightarrow \mathbb{R}$ (or, equivalently,
$\gamma(Q)$ is a Legendrian submanifold of $(T^*Q \times \mathbb{R}, \eta_Q)$), then
$$
\gamma_i  = \frac{\partial \gamma_z}{\partial q^i}
$$ 
and so~\eqref{11b} is fulfilled and~\eqref{11a} becomes
\begin{equation}\label{nuevo4}
d (H \circ \gamma) = 0.
\end{equation}

\begin{remark}
{\rm Notice that f and $\gamma_z$ define (locally) the same 1-jet.}
\end{remark}

Therefore, we have the following.

\begin{theorem}
 Assume that a section $\gamma$ of the projection $T^*Q \times \mathbb{R} \to Q$
is such that $\gamma(Q)$ is a Legendrian submanifold of $(T^{*}Q\times \mathbb{R}, \eta_Q)$. 
Then, the vector fields $\mathcal{E}_H$ and $\mathcal{E}_H^{\gamma}$ are $\gamma$-related if and only if (\ref{nuevo4}) holds. 
\end{theorem}

Equation (\ref{nuevo4}) is referred as a {\it Hamilton--Jacobi equation for the evolution vector field}.
A section $\gamma$ fullfilling the assumptions of the theorem and the Hamilton-Jacobi equation will
be called a {\it solution} of the Hamilton--Jacobi problem for the evolution vector field of $H$.


\subsubsection{Complete solutions}

As in the case of the Hamiltonian vector field, we can consider complete solutions
for the evolution vector field.

\begin{definition}
 A {\it complete solution} of the Hamilton--Jacobi equation for the evolution vector field
$\mathcal{E}_H$ of a Hamiltonian $H$ on a contact manifold $(M,\eta)$ 
 is a diffeomorphism $\Phi:Q\times \mathbb{R}\times \mathbb{R}^n\rightarrow T^{*}Q\times \mathbb{R}$ such that for any set of
 parameters $\lambda= (\lambda_0, \lambda_1, \dots , \lambda_n) \in \mathbb{R} \times \mathbb{R}^n$, the mapping
 
 \begin{equation}
 \begin{array}{ccc}
  \Phi_{\lambda}:Q & \rightarrow &  T^{*}Q\times \mathbb{R}  \\
  (q^i) &\mapsto &  \Phi_\lambda(q^i) = \Phi(q^i, \lambda_0, \lambda_1, \dots, \lambda_n)
 \end{array}
 \end{equation}
\noindent
is a solution of the Hamilton--Jacobi equation. 
\end{definition}

For simplicity, we will use the notation $(\lambda_\alpha \; , \; \alpha= 0, 1, \dots, n)$.

As in the previous case,
we define functions $f_\alpha$ such that for a point $p\in T^{*}Q\times \mathbb{R}$, it is satisfied
\begin{equation}\label{functions}
 f_\alpha(p) = \pi_\alpha  \circ \Phi^{-1}(p).
\end{equation}
where $\pi_\alpha: Q\times \mathbb{R}\times \mathbb{R}^n\rightarrow \mathbb{R}$ is the canonical projection
onto the $\alpha$ factor.

A direct computation shows that

$$
\hbox{Im} \; \Phi_\lambda = \cap_{\alpha=0}^n \, f_\alpha^{-1}(\lambda_\alpha)
$$
In other words,
$$
\hbox{Im} \; \Phi_\lambda = \{ x \in T^*Q \times \mathbb{R} \; | \; f_\alpha(x) = \lambda_\alpha, \alpha=0, \cdots, n\}
$$
Therefore, since under our hypthesis, $\mathcal{E}_H$ is tangent to any of the submanifolds
$\hbox{Im} \; \Phi_\lambda$, we deduce that
$$
\mathcal{E}_H (f_\alpha) = 0
$$
So, these functions are conserved quantities for the evolution vector field.

Moreover, we can compute
$$
\{f_\alpha, f_\beta\} = \Lambda (df_\alpha, df_\beta) - f_\alpha \mathcal{R}(f_\beta) + f_\beta \mathcal{R}(f_\alpha)
$$
But
$$
\Lambda (df_\alpha, df_\beta) = \sharp_\Lambda( df_\alpha)(f_\beta) = 0
$$
since $(T \hbox{Im} \Phi_\lambda)^\perp = T \hbox{Im} \Phi_\lambda$, so 
\begin{equation}\label{involution}
\{f_\alpha, f_\beta\} = - f_\alpha \mathcal{R}(f_\beta) + f_\beta \mathcal{R}(f_\alpha)
\end{equation}

\begin{theorem}
 There exist no linearly independent commuting set of first-integrals in involution \eqref{functions} for a complete solution of the Hamilton--Jacobi equation for the evolution vector field. 
\end{theorem}

{\bf Proof:} Since the images of the sections are Legendrian then they are
integral submanifolds of $\ker \, \eta_Q$. So, the Reeb vector field $\mathcal{R}$ will be
transverse to them, and consequently, there is at least some index $\alpha_0$ such that
$$
\mathcal{R} (f_{\alpha_0}) \not= 0
$$

Therefore, if all the brackets $\{f_\alpha, f_\beta\} $ vanish, then we woul obtain
that the functions $f_\alpha$ cannot be linearly independent. 

$\hfill \Box$

\section{Symplectification of the Hamilton-Jacobi equation}

\subsection{Homogeneous Hamiltonian systems and contact systems}\label{sec:homogenization}
There is a close relationship between homogeneous symplectic and contact systems, see for example~\cite{Iba,vds-2018}. Here we briefly recall some facts about the symplectification of cotangent bundles.

For any manifold $M$ a function $F:T^* M\to\mathbb{R}$ is said to be \textit{homogeneous} if, for any $p_q\in T^*_p M$, we have $F(\lambda p_q)=\lambda F(p_q)$ for any $\lambda \in \mathbb{R}$. In this situation the function $F$ can be projected to the projective bundle $\mathcal{P}(T^*M)$ over $M$ obtained by projectivization of every cotangent space. We are interested in the case that $M = Q\times\mathbb{R}$, with natural coordinates $(q^i,z,P_i,P_z)$ on $T^*(Q \times \mathbb{R})$. We note that this definition can be generalized to any vector bundle.

Let $\tilde{H}$ be an homogeneous Hamiltonian function on $T^*( Q \times \mathbb{R})$. Locally, we have that $\tilde{H}(q^i,z, \lambda P_i,  \lambda P_z) = \lambda \tilde{H}(q^i,z,P_i,P_z)$, for all $\lambda \in \mathbb{R}$. Equivalently, one can write
\begin{equation}
	\tilde{H}(q^i,z,P_i,P_z) = -P_z\,  H(q^i,-P_i/P_z, z),
\end{equation}
for $P_z\neq 0$, where $H: T^*Q \times \mathbb{R} \to \mathbb{R}$, $H(q^i,p_i,z)= \tilde{H}(q^i,z,p_i,-1)$ is well defined.

With the above changes, we have identified the manifold $T^*Q \times \mathbb{R}$ as the projective bundle $\mathcal{P}(T^* (Q \times \mathbb{R}))$ of the cotangent bundle $T^*(Q \times \mathbb{R})$ taking out the points at infinity,  that is the subset defined by $\{P_z = 0\}$.

Following \cite[Section~4.1]{vds-2018}, the map
\begin{equation}\label{dehomogeneization}
  \begin{aligned}
    \Phi: T^*( Q \times \mathbb{R}) \setminus \{P_z = 0 \} &\to T^*Q \times \mathbb{R}\\
    (q^i,z,P_i,P_z) &\mapsto (q^i, -P_i/P_z, z) = (q^i, p_i ,z),
  \end{aligned}
\end{equation}
sends the Hamiltonian symplectic system $(T^*( Q \times \mathbb{R})\setminus \{P_z = 0 \} , \omega_{Q \times \mathbb{R}}, \tilde{H})$ onto the Hamiltonian contact system $(T^*Q \times \mathbb{R}, \eta_{Q}, H)$, where $\omega_{Q \times \mathbb{R}} = d q^i \wedge d P_i + d z \wedge d P_z$ and $\eta_Q = d z - p_i d q^i$ are the canonical symplectic and contact forms, respectively. Observe that the natural coordinates of $T^*Q\times\mathbb{R}$, denoted by $(q^i,p_i,z)$, correspond to the homogeneous coordinates in the projective bundle. In fact, the map $\Phi$ is  projectivization up to a minus sign; i.e., the map that sends each point in the fibers of $T^* (Q \times \mathbb{R})$ to the line that passes through it and the origin. 

The map $\Phi$ satisfies $\bar{H} = -P_z\Phi^* (H)$ and $\omega_Q = -d (P_z \Phi^* (\eta_Q))$

It can be shown that $\Phi$ provides a bijection between conformal contactomorphisms and homogeneous symplectomorphisms. Moreover, $\Phi$ maps homogeneous Lagrangian submanifolds $\mathcal{L} \subseteq T^*( Q \times \mathbb{R})$ onto Legendrian submanifolds $\mathbb{L} = \Phi(\mathcal{L}) \subseteq T^*Q \times \mathbb{R}$. Indeed, if $\mathcal{L}$ is homogeneous, then $\mathbb{L}$ is Legendrian if and only if $\mathcal{L}$ is Lagrangian. Moreover, the Hamilton equations for $\tilde{H}$ are transformed into the Hamilton equations for $H$, i.e., $\Phi_* X_{\tilde{H}} = X_H$. See~\cite{vds-2018,vds-2018a} for more details on this topics.

We also remark that this construction is symplectomorphic to the symplectification defined in~\cite{Iba}, which is given by
$$
	(T^*Q \times \mathbb R \times \mathbb R,\omega=e^{-t}(d\eta_Q+\eta_Q \wedge dt) = d(e^{-t} \eta_Q))
$$
where $t$ is the (global) coordinate of the second $\mathbb R$ factor with the ``symplectified'' Hamiltonian $\tilde{H}' = e^t H$  setting and then project it to the original contact manifold. That is, $\tilde{H}'$ such that
\begin{equation}\label{eq:sympl_hamiltonian}
	(pr_1)_*  X_{\tilde{H'}} =X_H ,
\end{equation}
where $pr_1 \colon T^*Q \times \mathbb R \times \mathbb R \to T^*Q \times \mathbb R$ is the projection onto the first two factors.

The following map provides the symplectomorphism
\begin{equation}
	\begin{split}
		\Psi: T^*(Q \times \mathbb{R})\cap \{P_z < 0 \} & \to T^*Q \times \mathbb{R} \times \mathbb{R}\\
		(q^i,z,P_i,P_z) &\to (q^i, -P_i/P_z, z, -\log (-P_z)) = (q^i, p_i ,z),
	\end{split}
\end{equation}
that is, $\Psi = (\Phi, -\log(-P_z)$. This map is a symplectomorphism that maps $\bar{H}$ onto $\bar{H}'$. Moreover it is a fiber bundle automorphism over $TQ\times\mathbb{R}$ (see the diagram below):
\begin{equation}
	\begin{tikzcd}
		T^*(Q \times \mathbb{R})\cap \{P_z < 0 \} \arrow[rr, "\Psi"] \arrow[rd, "\Phi"] &                      & T^*Q \times \mathbb{R} \times \mathbb{R} \arrow[ld, "pr_1"] \\
																						& TQ \times \mathbb{R} &                                                            
		\end{tikzcd}
\end{equation}


\subsection{Relations for the Hamilltonian vector field}

Now we will establish a relationship between solutions to the Hamilton–Jacobi problem in both scenarios. Suppose that
\begin{align*}
\tilde{\gamma}\colon Q\times\mathbb R &\to T^*(Q\times\mathbb R)\\
(q^i,z)&\mapsto(q^i,\tilde{\gamma}_j(q^i,z), z,\tilde{\gamma}_t(q^i, z))
\end{align*}
 is a solution of the symplectic Hamilton-Jacobi equation, i.e., $\tilde{\gamma}(Q\times\mathbb R)$ is Lagrangian and
$$
d(\tilde{H}\circ\tilde{\gamma})=0,
$$
or equivalently
$$
T\tilde{\gamma}\circ X_{\tilde{H}}^{\tilde{\gamma}}=X_{\tilde{H}}\circ \tilde{\gamma},
$$
where $X_{\tilde{H}}^{\tilde{\gamma}}=Tp\circ X_{\tilde{H}}\circ \tilde{\gamma}$ is the projected vector field and $p\colon T^*(Q\times\mathbb R)\to Q\times\mathbb R$ the canonical projection.   We want to use the solution $\tilde{\gamma}$ of the Hamilton-Jacobi problem in the symplectification (which we will often refer to as ``symplectic solution'') to obtain a section that is a solution in the contact setting (``contact solution'', for simplicity). We assume $\tilde{\gamma}_t(q^i, z) \neq 0$ and take $\gamma=\Phi\circ \tilde{\gamma} \colon Q\times\mathbb R \to T^*Q\times\mathbb R$. In local coordinates
\begin{align*}
	\gamma\colon Q\times\mathbb R &\to T^*Q\times\mathbb R\\
	(q^i,z)&\mapsto\left(q^i, \gamma_{j}(q^i,z) =-\frac{\tilde{\gamma}_j(q^i,z)}{\tilde{\gamma}_t(q^i, z)}, z\right)
\end{align*}

We can summarize the situation in the following commutative diagram:
\begin{equation}\label{symplectification_diagram}
	\begin{tikzcd}[row sep=1.4cm]
		Q\times\mathbb R \arrow[r, "\tilde{\gamma}"] \arrow[d, "X_{\tilde{H}}^{\tilde{\gamma}} = X_{{H}}^{{\gamma}}"'] \arrow[rr, "\gamma", bend left=49] & T^*(Q\times\mathbb R) \arrow[r, "\Phi"] \arrow[d, "X_{\tilde{H}}"] \arrow[l, "p", bend left] & T^*Q\times\mathbb R \arrow[d, "X_H"] \arrow[ll, "\pi", bend right] \\
		T(Q\times\mathbb R) \arrow[r, "T\tilde{\gamma}", dashed] \arrow[rr, "T \gamma"', dashed, bend right=49]                                          & T(T^*(Q\times\mathbb R)) \arrow[r, "T\Phi"] \arrow[l, "Tp", bend right]                      & T(T^*Q\times\mathbb R) \arrow[ll, "T\pi"', bend left]             
		\end{tikzcd}
	\end{equation}
We note that the projected vector fields $X_{\tilde{H}}^{\tilde{\gamma}}$ and $X_{{H}}^{{\gamma}}$ coincide. The dashed lines of $T \tilde{\gamma}$ (resp. $T \gamma$) commute if and only if $\tilde{\gamma}$ is a symplectic solution (resp. ${\gamma}$ is a contact solution) of the Hamilton-Jacobi problem.

	\begin{lemma}
	Let $H$ be a Hamiltonian and $\tilde{H}$ its symplectified version. Assume $\tilde{\gamma}_t(q^i, z)\neq 0$. Then $\tilde{\gamma}$ is a symplectic solution, or, equivalently, $X^{\tilde{\gamma}}_{\tilde{H}}$ and $X_{\tilde{H}}$ are $\tilde{\gamma}$-related if and only if $X^{{\gamma}}_{{H}}$ and $X_H$ are $\gamma$-related.
	\end{lemma}

    \begin{proof}

    Assume that  $X^{\tilde{\gamma}}_{\tilde{H}}$ and $X_{\tilde{H}}$ are $\tilde{\gamma}$-related. Then, by the commutativity of the diagram~\eqref{symplectification_diagram} we see that $X^{{\gamma}}_{{H}}$ and $X_H$ are $\gamma$-related.

	Conversely, assume that $X^{{\gamma}}_{{H}}$ and $X_H$ are $\gamma$-related. Let $P_z \in \mathbb{R} \setminus \{ 0 \}$ and let
	\begin{equation}
\begin{split}
			\xi: T^* Q \times \mathbb{R} &\to T^*(Q \times \mathbb{R})\\
			(q^i,P_i,z) &\mapsto (q^i, z, -P_i \tilde{\gamma}_t(q^i,z))	
\end{split}	
\end{equation}

    We note that $\xi_{P}$ is the inverse of $\Phi$ along the submanifold $\{P_z = \gamma_t\} \subseteq T^*(Q \times \mathbb{R})$. In particular $\tilde{\gamma} = \xi \circ \gamma$. Looking at the diagram~\eqref{symplectification_diagram}, this implies that  $X^{\tilde{\gamma}}_{\tilde{H}}$ and $X_{\tilde{H}}$ are $\tilde{\gamma}$-related. 
\end{proof}

\begin{lemma}
    Assume that the image of $\tilde{\gamma} = (\tilde{\gamma}_Q,\tilde{\gamma}_t)$ is Lagrangian. Then the image of $\gamma$ is coisotropic and the images of $\gamma_{z_0}$ are Lagrangian if and only if $d_Q \tilde{\gamma_Q} = \tau \gamma_Q$ for some function $\tau:Q\times\mathbb{R} \to \mathbb{R}$.
    
    Conversely, if the image of $\gamma$ is coisotropic and the images of $\gamma_{z_0}$ are Lagrangian, then we can choose $\tilde{\gamma_t}$ so that the image of $\tilde{\gamma}$ is coisotropic. It is given by $\tilde{\gamma}_t = \pm \exp(g)$, where $g$ is a solution to the PDE
    \begin{equation}
        d_Q g + \gamma \lieD{{\partial }/{\partial z}} g = -\lieD{\partial / \partial z} \gamma.
        \end{equation}
\end{lemma}\label{lemma_lift_submanifolds}
    \begin{proof}
    Let $\tilde{\gamma} = (\tilde{\gamma}_Q,\tilde{\gamma}_t)$ be such that its image is Lagrangian. That is,  $d \tilde{\gamma} = 0 $. Splitting the part in $Q$ and in $\mathbb{R}$, we see that this is equivalent to 
    \begin{equation}\label{tildegamma_lagrangian}
        \lieD{\partial/\partial z} \tilde{\gamma}_Q = d_Q \tilde{\gamma}_t, \quad d_Q \tilde{\gamma}_Q = 0.
    \end{equation}

    Now, $\gamma = -\tilde{\gamma}_Q / \tilde{\gamma}_t$. By Theorem~\ref{thm:coisotropic_lagrangian_section}, it is necessary that $d_Q\gamma = 0$ and $(\lieD{\partial/\partial z}  \gamma )\wedge \gamma = 0$. We compute
        \begin{equation}
            d_Q \gamma =  \frac{(d_Q\tilde{\gamma}_t) \wedge \tilde{\gamma}_Q}{\tilde{\gamma}_t^2} =
            \frac{(\lieD{\partial/\partial z}\tilde{\gamma}_Q) \wedge \tilde{\gamma}_Q}{\tilde{\gamma}_t^2} =
        \end{equation}
        \begin{equation}
            \begin{split}                
            (\lieD{\partial/\partial z}  \gamma ) \wedge \gamma &=  
            -\frac{(\lieD{\partial/\partial z}\tilde{\gamma}_t) \tilde{\gamma}_Q - 
            (\lieD{\partial/\partial z}\tilde{\gamma}_Q) \tilde{\gamma}_t
            }{\tilde{\gamma}_t^2} \wedge \frac{\tilde{\gamma_Q}}{\tilde{\gamma}_t} \\ &= 
               \frac{ 
            (\lieD{\partial/\partial z}\tilde{\gamma}_Q) \wedge \tilde{\gamma}_Q
            }{\tilde{\gamma}_t^2}
        \end{split}
        \end{equation}
        hence the images of $\gamma_{z_0}$ are Lagrangian and the image of $\gamma$ is coisotropic if and only if $\lieD{\partial/\partial z} (\tilde{\gamma}_Q)$ is proportional to $\tilde{\gamma}_Q$.

        Conversely, assume that $\gamma$ satisfies $d_Q \gamma = 0$ and $\lieD{\partial / \partial z} \gamma = \sigma \gamma$. We must find $\tilde{\gamma}_t$ so that~\eqref{tildegamma_lagrangian} are satisfied. Since $\tilde{\gamma}_Q = - \tilde{\gamma}_t \gamma$, we have that~\eqref{tildegamma_lagrangian} are equivalent to
        \begin{align}
            \lieD{{\partial }/{\partial z}} ( \tilde{\gamma}_Q ) &= -  (\lieD{{\partial }/{\partial z}} \tilde{\gamma}_t  + \sigma \tilde{\gamma}_t ) \wedge \gamma = d_Q \tilde{\gamma}_t \label{lift_submanifold_eq}\\
            d_Q (\gamma_Q) &= - d_Q \tilde{\gamma}_t \wedge \gamma = 0.
        \end{align}
        A solution for $\tilde{\gamma}_t $ on the first equation above clearly solves the second one. Since we look for nonvanishing $\tilde{\gamma}_t$,
        we let $g = \log \circ \rvert{\tilde{\gamma}_t}\lvert$ so that~\label{lift_submanifold_eq} is just
        \begin{equation}
            d_Q g + \gamma \lieD{{\partial }/{\partial z}} g = -\sigma \gamma = -\lieD{\partial / \partial z} \gamma,
        \end{equation}
        if we let
        \begin{equation}
            A_i = \frac{\partial }{\partial q^i} + \gamma^i \frac{\partial }{\partial z},
        \end{equation}
        this equation can be written as 
        \begin{equation}
            A_i(g) = -\frac{\partial \gamma_i}{\partial z}, 
        \end{equation}
        we note that this vector fields commute, indeed
        \begin{equation}
            \lbrack X_i,X_j\rbrack = \gamma_i \frac{\partial \gamma_j}{\partial z} -  \gamma_j \frac{\partial \gamma_i}{\partial z} = \sigma \gamma_i \gamma_j - \sigma \gamma_j \gamma_i = 0.
        \end{equation}
        If this PDE has local solutions, operating with the equations above, one has,
        \begin{equation*}
            A_i(\frac{\partial \gamma_j}{\partial z}) - A_j(\frac{\partial \gamma_i}{\partial z}) = 0.
        \end{equation*}
        This condition is clearly neccesary, and it is also sufficient by~\cite[Thm.~19.27]{leeManifolds}.  We have that
        \begin{equation}\label{lift_submanifold_cond}
            A_i(\frac{\partial \gamma_j}{\partial z}) - A_j(\frac{\partial \gamma_i}{\partial z})  =
            \gamma_i \frac{\partial^2 \gamma_j}{(\partial z)^2} -  \gamma_j \frac{\partial^2 \gamma_i}{(\partial z)^2} =
            \gamma_i \frac{\partial (\sigma \gamma_j)}{\partial z} -  \gamma_j \frac{\partial (\sigma \gamma_i)}{\partial z} = 0.
        \end{equation}
    \end{proof}

    Combining the last two results, we obtain a correspondence between symplectic and contact solutions to the Hamilton–Jacobi problem.
    \begin{theorem}\label{lemma_lift_submanifolds}
        Let $H$ be a Hamiltonian and $\tilde{H}$ its symplectified version. Then $\tilde{\gamma}\colon Q\times\mathbb R\to T^*(Q\times\mathbb R)$ is a solution of the symplectic Hamilton-Jacobi problem for $\tilde{H}$, if and only if $\gamma= \Phi\circ\tilde{\gamma}\colon Q\times\mathbb R \to T^*Q\times \mathbb R$ is a solution of the contact Hamilton-Jacobi problem for $H$ and $d_Q \tilde{\gamma}_Q = \tau \gamma_Q$ for some function $\tau:Q\times\mathbb{R} \to \mathbb{R}$.

        Conversely, given a contact solution $\gamma$ of the Hamilton-Jacobi equation, there exists a symplectic solutions $\tilde{\gamma}$ such that $\tilde{\gamma} = \pm \exp(g)$, where $g$ is a solution to the PDE
        \begin{equation}
            d_Q g + \gamma \lieD{{\partial }/{\partial z}} g = -\lieD{\partial / \partial z} \gamma.
            \end{equation}
    \end{theorem}

\subsubsection{The alternative approach}

For each $z$, we have sections $\gamma=pr_2\circ\tilde{\gamma}_z\colon Q\to T^*Q\times\mathbb R$ of the form $(q^i)\mapsto (q^i, \tilde{\gamma}_j(q^i,z),\tilde{\gamma}_t(q^i,z))$, being $pr_2:(q^i,p_i,z,t)\mapsto(q^i, p_i, t)$. We know that $\gamma$ is a solution of the contact Hamilton-Jacobi problem if and only if $\gamma(Q)$ is Legendrian and
	$$
	H\circ\gamma=0.
	$$
	The condition that $\gamma(Q)$ is Legendrian is equivalent to
	$$
	\gamma_i=\fracpartial{\gamma_z}{q^i},
	$$
	where we write $\gamma(q^i)=(q^i,\gamma_j(q^i),\gamma_z(q^i))$, which by definition of $\gamma$  and using that $\gamma(Q\times\mathbb R)$ is Lagrangian reads
	$$
	\tilde{\gamma}_i=\fracpartial{\tilde{\gamma}_t}{q^i}=\fracpartial{\tilde{\gamma}_i}{z}
	$$
	and therefore $\tilde{\gamma}_i=e^zg(q^i)$, with $g_i$ functions depending only on the $(q^i)$. This can be summarized as follows:
	\begin{theorem}
	Suppose $\tilde{\gamma}\colon Q\times\mathbb R \to T^*(Q\times\mathbb R)$ is a solution of the symplectified Hamilton-Jacobi problem. Then 
	\begin{align*}
	\gamma\colon Q &\to T^*Q\times\mathbb R\\
	(q^i)&\mapsto (q^i, \tilde{\gamma}_j(q^i,z), \tilde{\gamma}_t(q^i,z))
	\end{align*}
	is a solution of the contact Hamilton-Jacobi problem if and only if
	$$
	H\circ\gamma=0 \text{ and } \tilde{\gamma}_i=e^zg_i
	$$
	\end{theorem}
	
\subsection{Relations for the evolution vector field}
	
	We now consider the evolution field $\mathcal{E}_H$.  First, note that 
	$$
	Tpr_1\circ X_{\tilde{H}}=(\mathcal E_H-H\Reeb)\circ pr_1
	$$
	so that we cannot simply expect to project the vector field as before.  In fact, one can easily prove that under the assumption that the symplectified Hamiltonian is of the form
	$$
	\tilde{H}=F(t,H),
	$$
	then the associated vector field $X_{\tilde{H}}$ such that $i_{X_{\tilde{H}}}\omega=dH$ will never verify 
	$$
	Tpr_1\circ X_{\tilde{H}}=\mathcal E_H\circ pr_1
	$$
	
	We will now see that, despite this apparent obstruction, one can still establish some relations. Let $\tilde{\gamma}\colon Q\times\mathbb R\to T^*(Q\times\mathbb R)$ be a solution of the symplectified problem and define the section $\gamma=pr_2\circ\tilde{\gamma}_z\colon Q\to T^*Q\times\mathbb R$. This will be a solution of the associated Hamilton-Jacobi problem for the evolution field if and only if $\gamma(Q)$ is Legendrian and
	$$
	d(H\circ\gamma)=0.
	$$
	The Legendrian condition is equivalent to 
	$$
	\gamma_i=\fracpartial{\gamma_t}{q^i},
	$$
	or, using that $\tilde{\gamma}(Q\times\mathbb R)$ is Lagrangian like in the previous section,
	$$
	\tilde{\gamma}_i=e^zg_i(q^i).
	$$
	On the other hand, we know that $\tilde{\gamma}$ is a solution of the symplectic problem and therefore $d(\tilde{H}\circ\tilde{\gamma})=0$, which by definition means
	$$
	e^{-\tilde{\gamma}_t(q^i,z)}H(q^i,\tilde{\gamma}_j(q^i,z),z)=C
	$$
	with $C$ constant. Since $\gamma(q^i)=(q^i,\tilde{\gamma}_j(q^i,z),\tilde{\gamma}_t(q^i,z)$),  using the previous equation we obtain:
	$$
	H\circ\gamma=Ce^{\tilde{\gamma}_t(-.\tilde{\gamma}_t(-,z))}.
	$$
	Then the condition $d(H\circ\gamma)=0$ reads
	$$
	(H\circ\gamma)\left(\fracpartial{\tilde{\gamma}_t}{q^i}+\fracpartial{\tilde{\gamma}_t}{z}\fracpartial{\tilde{\gamma}_t}{q^i}\right)=0
	$$
	which occurs if and only if at every point $(q^i)$ we have:
	$$
	H\circ\gamma=0 \text{ or }\fracpartial{\tilde{\gamma}_t}{q^i}=0 \text{ or } \fracpartial{\tilde{\gamma}_t}{z}=-1.
	$$
	The functional form found for $H\circ\gamma$ tells us that it is either non-zero at every point or it vanishes everywhere.  If it does not vanish (everywhere), we claim that the second equation must be true.  Indeed,  suppose the first two equations do not hold.  Then the third equation must be true not just at a given point but in an open neighborhood and we would have
	$$
	\tilde{\gamma}_t=-z+h(q^i),
	$$
	where $h_i$ are arbitrary functions.   Using again that $\tilde{\gamma}(Q\times\mathbb R)$ is Lagrangian we could write
	$$
	\fracpartial{\tilde{\gamma}_t}{q_i}=\fracpartial{h}{q^i}=e^zg_i(q^i)=\fracpartial{\tilde{\gamma}_i}{z}
	$$
	which would imply that $h$ depends also on $z$.  Therefore, if $H\circ\gamma\neq 0$ then the second equation is true at every point. Using that $\tilde{\gamma}(Q\times\mathbb R)$ is Lagrangian we see this is equivalent to $\tilde{\gamma}_i=0$. Therefore we find:
	\begin{theorem}
	Let $\tilde{\gamma}\colon Q\times\mathbb R \to T^*(Q\times\mathbb R)$ be a solution of the symplectified problem with $\tilde{\gamma}_i=e^z g_i$, where $g_i:Q\to\mathbb R$, and consider the section 
	\begin{align*}
	\gamma\colon Q &\to T^*Q\times\mathbb R\\
	(q^i) & \mapsto (q^i,\tilde{\gamma}_j(q^i,z), \tilde{\gamma}_t(q^i,z))
	\end{align*}
	Then $\gamma$ is a solution of the contact problem for the evolution field if and only if one of the two following conditions is fulfilled:
	\begin{enumerate}
	\item $H\circ\gamma=0$,
	\item $\tilde{\gamma}_i=0$.
	\end{enumerate}

	\end{theorem}

\section{Examples}
\subsection{Particle with linear dissipation}
Consider the Hamilltonian $H$
\begin{equation}
    H(q,p,z) = \frac{p^2}{2m} + V(q) + \lambda z,
\end{equation}
where $\lambda \in \mathbb{R}$ is a constant. The extended phase space is $T^*Q \times \mathbb{R} \simeq \mathbb{R}^3$.

The Hamiltonian and evolution vector field are given by
\begin{align}
    X_H &= \frac{p}{m} \frac{\partial }{\partial q} 
    -\left( \frac{\partial V}{\partial q} + \lambda z \right) \frac{\partial }{\partial p} +
    \left( \frac{p^2}{2m} -V(q) -\lambda z \right) \frac{\partial }{\partial z}, \\
    \mathcal{E}_H &= \frac{p}{m} \frac{\partial }{\partial q} 
    -\left( \frac{\partial V}{\partial q} + \lambda z \right) \frac{\partial }{\partial p} +
    \frac{p^2}{m} \frac{\partial }{\partial z}.
\end{align}

Assume that $\gamma : Q \to T^* Q \times \mathbb{R}$ is a section of the canonical projection $T^* Q \times \mathbb{R} \to Q$, that is,
\begin{equation}
    \gamma(q) = (q,\gamma_p(q), \gamma_z(q)).
\end{equation}

We assume that $\gamma(Q)$ is a Legendrian submanifold of $T^*Q \times \mathbb{R}$ as in Section~\ref{sec:evolution_alt}; then,
\begin{equation}
    \gamma_p(q) = \frac{\partial \gamma_z}{\partial q},
\end{equation}
and $\mathcal{E}_H$ and $\mathcal{E}^\gamma_H$ are $\gamma$-related if and only if
\begin{equation}
    H \circ \gamma = k,
\end{equation}
for a constant $k \in \mathbb{R}$. Then, the Hamilton–Jacobi equation becomes
\begin{equation}
    H(\gamma(q)) = \frac{\gamma_p^2}{2m} + V(q) + \lambda \gamma_z = k,
\end{equation}
or, equivalently,
\begin{equation}\label{eq:hjex1}
    \frac{\left( \frac{\partial \gamma_z}{\partial q} \right)^2}{2m} + V(q) + \lambda \gamma_z = k,
\end{equation}
which is a linear ordinary differential equation.

A general solution of the Hamilton–Jacobi equation~\eqref{eq:hjex1} is then 
\begin{equation}
    \gamma_p(q) = \exp(-2m\lambda q) \int (2mk - 2mV(q)) \exp(2m\lambda q) dq.
\end{equation}

\subsection{Application to thermodynamic systems}\label{apthersys}

We consider thermodynamic systems in the so called \emph{energy representation}. Hence
the \emph{thermodynamic phase space}, representing the extensive variables,  is the manifold $T^* Q  \times \mathbb{R}$, equipped with its canonical contact form
\begin{equation}
  \eta_Q = d U - p^i d q^i.
\end{equation}
The local coordinates on the configuration manifold $Q$ are $(q^i,U)$, where $U$ is the internal energy and $q^i$'s denote the rest of extensive variables.
Other variables, such as the entropy, may be chosen instead of the internal energy, by means of a Legendre transformation.

The state of a thermodynamic system always lies on the equilibrium submanifold $\mathbb{L}\subseteq T^* Q  \times \mathbb{R}$, which is a Legendrian submanifold. The pair $(T^* Q  \times \mathbb{R}, \mathbb{L})$ is a \emph{thermodynamic system}. The equations (locally) defining $\mathbb{L}$ are called the \emph{state equations} of the system.

On a thermodynamic system $(T^* Q  \times \mathbb{R}, \mathbb{L})$, one can consider the dynamics generated by a Hamiltonian vector field $X_H$ associated to a Hamiltonian $H$. If this dynamics represents \emph{quasistatic processes}, meaning that at every time the system is in equilibrium, that is, its evolution states remain in the submanifold $\mathbb{L}$, it is required for the contact Hamiltonian vector field $X_H$ to be tangent to $\mathbb{L}$. This happens if and only if $H$ vanishes on $\mathbb{L}$.

Using Hamilton-Jacobi theory, one sees that a section $\gamma$ satisfied $H \circ \gamma = 0$ if and only if $X^\gamma_H$ and $X_H$ are $\gamma$-related.


\subsubsection{The classical ideal gas}
This example is fully described in~\cite{goshTermo}.

The classical ideal gas is described by the following variables.
\begin{itemize}
    \item $U$: internal energy,
    \item $T$: temperature,
    \item $S$: entropy,
    \item $P$: pressure,
    \item $V$: volume,
    \item $\mu$: chemical potential,
    \item $N$: mole number.
\end{itemize}
Thus, the thermodynamic phase space is $T^* \mathbb{R}^3 \times \mathbb{R}$ and the contact $1$-form is
\begin{equation}
    \eta = d U - T d S + P d V - \mu d N
\end{equation}

The Hamilltonian function is
\begin{equation}
    H = T S - R N T + \mu N - U,
\end{equation}
where $R$ is the constant of ideal gases. The Reeb vector field is $\Reeb = \frac{\partial }{\partial U}$.

The Hamilltonian and evolution vector fields are just
\begin{align}
    X_H &= (S-RN)\fracpartial{}{S} + N \fracpartial{}{N}+P\fracpartial{}{P} + RT\fracpartial{}{\mu} + U \fracpartial{}{U} \\
    \mathcal{E}_H &=(S-RN)\fracpartial{}{S} + N \fracpartial{}{N}+P\fracpartial{}{P} + RT\fracpartial{}{\mu} + (TS-RNT+\mu N) \fracpartial{}{U}.
\end{align}
The Hamilonian vector field here represents an isochoric and isothermal process on the ideal gas.

Assume that $\gamma: \mathbb{R}^3 \to T^* \mathbb{R}^3 \times \mathbb{R}$ is the section locally given by
\begin{equation}
    \gamma(S,V,N) = (S, V, N, \gamma_T, \gamma_P, \gamma_\mu \gamma_U).
\end{equation}
we know that $\gamma(\mathbb{R}^3)$ is a Legendrian submanifold of $(T^* \mathbb{R}^3 \times \mathbb{R}, \eta)$ if and only if,
\begin{align*}
    \gamma_T &= \frac{\partial \gamma_U}{\partial S},\\
    \gamma_P &= - \frac{\partial \gamma_U}{\partial V},\\
    \gamma_\mu &= \frac{\partial \gamma_U}{\partial N}.
\end{align*}

The Hamilton–Jacobi equation is
\begin{equation}
    (H \circ \gamma)(S,V,N) = (S - R N)\gamma_T + N \gamma_\mu - \gamma_U = k,
\end{equation}
for some $k \in \mathbb{R}$. That is,
\begin{equation}
    (H \circ \gamma) (S,V,N) = (S - R N) \frac{\partial \gamma_U}{\partial S} + N \frac{\partial \gamma_U}{\partial N} - \gamma_U = k.
\end{equation}
This is a first order linear PDE, whose solution is given by
\begin{equation}
    \gamma_U  (S,V, N) = k \operatorname{arcsinh}{\left (\frac{S}{\sqrt{- S^{2} + \left(- 7 N + S\right)^{2}}} \right )} + F(- S^{2} + (R N - S)^{2} , V),
\end{equation}
with $F:\mathbb R^2\to \mathbb R$ an arbitrary function. The case $k=0$, which is the one relevant for the thermodynamic interpretation, is given by
\begin{equation}
    \gamma_U  (S,V, N) = F(- S^{2} + (R N - S)^{2} , V).
\end{equation}

\section*{Acknowledgements}

We acknowledge the financial support from the MINECO Grant PID2019-106715GB-C21. Manuel Laínz wishes to thank MICINN and ICMAT for a FPI-Severo Ochoa predoctoral contract PRE2018-083203.  Álvaro Muñiz thanks ICMAT for the ``Grant Programme Severo Ochoa – ICMAT: Introduction to Research 2020'' and Fundación Barrié for its fellowship for postgraduate studies.

\end{document}